\documentclass[a4paper,leqno]{amsart}

\usepackage{amsmath}
\usepackage{amsthm}
\usepackage{amssymb}
\usepackage{amsfonts}
\usepackage{mathrsfs}
\usepackage{pdfsync}
\usepackage{hyperref}

\newcommand{\be}{\begin{equation}}
\newcommand{\ee}{\end{equation}}

\def\C{{\mathbb C}}
\def\R{{\mathbb R}}
\def\N{{\mathbb N}}
\def\Sch{{\mathcal S}}
\def\O{\mathcal O}

\def\({\left(}
\def\){\right)}
\def\<{\left\langle}
\def\>{\right\rangle}
\def\le{\leqslant}
\def\ge{\geqslant}

\def\Eq#1#2{\mathop{\sim}\limits_{#1\rightarrow#2}}
\def\Tend#1#2{\mathop{\longrightarrow}\limits_{#1\rightarrow#2}}

\def\d{{\partial}}
\def\eps{\varepsilon}

\def\U{{\mathcal U}}

\DeclareMathOperator{\RE}{Re}

\theoremstyle{plain}
\newtheorem{theorem}{Theorem}[section]
\newtheorem{lemma}[theorem]{Lemma}

\newtheorem{proposition}[theorem]{Proposition}
\newtheorem{hyp}[theorem]{Assumption}

\theoremstyle{definition}

\newtheorem{remark}[theorem]{Remark}
\newtheorem*{remark*}{Remark}
\newtheorem{example}[theorem]{Example}

\numberwithin{equation}{section}

\begin{document}

\title[Semiclassical wave packets in periodic potentials]
{Semiclassical wave packet dynamics in 
Schr\"odinger equations with
  periodic potentials}

\author[R. Carles]{R\'emi Carles}
\address[R. Carles]{Univ. Montpellier~2\\Math\'ematiques \\
  CC~051\\F-34095 Montpellier} 
\address{CNRS, UMR 5149\\  F-34095 Montpellier\\ France}
\email{Remi.Carles@math.cnrs.fr}
\author[C. Sparber]{Christof Sparber}
\address[C. Sparber]
{Department of Mathematics, Statistics, and Computer Science\\
University of Illinois at Chicago\\
851 South Morgan Street
Chicago, Illinois 60607, USA}
\email{sparber@uic.edu}

\begin{abstract}
We consider semiclassically scaled Schr\"odinger equations with an
external potential and a  highly
oscillatory periodic potential. We construct asymptotic solutions in
the form of  semiclassical wave packets. These solutions are
concentrated (both, in space and in frequency) around the effective
semiclassical phase-space flow obtained by Peierls substitution, and involve a 
slowly varying envelope whose dynamics is governed by a homogenized
Schr\"odinger equation with 
time-dependent effective mass. The corresponding adiabatic decoupling
of the slow and 
fast degrees of freedom is shown to be valid
up to Ehrenfest time scales.  
\end{abstract}


\thanks{R.~Carles is supported by the French ANR project
  R. A. S. (ANR-08-JCJC-0124-01). C.~Sparber has been supported by the Royal
Society via his University Research Fellowship}
\maketitle

\section{Introduction} 
\label{sec:intro}

\subsection{General setting}
We consider the following semiclassically scaled 
Schr\"odinger equation: 
\begin{equation}\label{eq:NLS0}
 \left\{
\begin{aligned}
     i\eps\d_t \psi^\eps +\frac{\eps^2}{2}\Delta \psi^\eps&=V_\Gamma
     \left(\frac{x}{\eps}\right)\psi^\eps+ V( x)\psi^\eps, \
(t,x)\in
    {\R}\times {\R}^d, \\
     \psi^\eps_{\mid t=0} &= \psi^\eps_{0},
  \end{aligned}
\right.
\end{equation}
with $d\ge 1$, the spatial dimension, and $\psi^\eps = \psi^\eps(t,x)\in \C$. Here, we already have rescaled all physical parameters such that only 
one semiclassical parameter $\eps >0$  (i.e. the scaled Planck's constant) remains. In the following we shall be interested in the asymptotic 
description of $\psi^\eps(t,x)$ for $\eps \ll 1$.
To this end, the potential
$ V_\Gamma(y)\in \R$ is assumed to be smooth and 
\emph{periodic} with respect to some \emph{regular lattice} $\Gamma \simeq
\mathbb Z^d$, generated by a given basis $\{\eta_1,\dots,\eta_d\}$, $\eta_\ell
\in \R^d$, i.e.
\begin{equation}
\label{eq:Vper}
V_{\Gamma}(y + \gamma) = V_{\Gamma}(y), \quad \forall \, y \in \R^d,
\gamma \in \Gamma
\end{equation}
where 
\begin{equation*}
\Gamma \equiv  
\Big\{\gamma=\sum_{\ell=1}^d \gamma_\ell
  \eta_\ell\in \R^d:  \: \gamma_\ell\in \mathbb Z \Big \}. 
\end{equation*} 
In addition, the slowly-varying potential $V$ is assumed to satisfy the following:
\begin{hyp}\label{hyp:V}
 $ 
 V\in C^3( \R^d;\R)$ and  $\d^\gamma V\in
 L^\infty\( \R^d\)$, for $ |\gamma|= 2,3$.  
\end{hyp}
Note that this implies that $V(x)$ grows at most quadratically at infinity.  
Equation \ref{eq:NLS0} describes the dynamics of quantum particles in a periodic lattice-potential $V_\Gamma$ 
under the influence of an external, \emph{slowly varying} driving force $F= - \nabla V(x)$.  
A typical application arises in solid state physics where
\eqref{eq:NLS0} describes the time-evolution of electrons moving in a
crystalline lattice (generated by the ionic cores).  
The asymptotics of \eqref{eq:NLS0} as $\eps \to 0_+$ is a natural
two-scale problem which is well-studied in the physics and mathematics
literature. Early mathematical results are based on time-dependent 
WKB type expansions \cite{BeLiPa, GuRaTr88, Spohn96} (see also
\cite{CaMaSp} for a more recent application in the nonlinear case),
which, however, suffer from the appearance of caustics and are thus
only valid for small times.  
In order to overcome this problem, other methods based on, e.g.,
Gaussian beams \cite{DiGuRa06}, or Wigner measures \cite{Ge91, GMMP}, have been
developed. These approaches yield an asymptotic description for time-scales of
order $\O(1)$ (i.e. beyond caustics). More recently, so-called space-adiabatic perturbation 
theory has been used (together with Weyl pseudo-differential calculus) to derive an effective Hamiltonian, 
governing the dynamics of particles in periodic potentials $V_\Gamma$ under the additional influence of slowly 
varying perturbations \cite{HoSpTe01, Te03}. The semi-classical asymptotics of this effective model is then 
obtained in a second step, invoking an Egorov-type theorem.

On the other hand, it is well known that in the case \emph{without periodic potential}, 
semiclassical approximations which are valid up
to \emph{Ehrenfest time} $t\sim \mathcal O(\ln 1/\eps)$ can be
constructed in a rather simple way.  
The corresponding asymptotic method is based on propagating \emph{semiclassical wave
  packets}, or coherent states, i.e. approximate solutions of \eqref{eq:NLS0} which are
sufficiently concentrated \emph{  
in space and in frequency} around the classical Hamiltonian phase-space flow. 
More precisely, one considers
\begin{equation}\label{eq:approx}
\psi^\eps(t,x) \approx   \eps^{-d/4} u 
\left(t,\frac{x-q(t)}{\sqrt\eps}\right)  e^{i \(S(t)+(x-q(t))\cdot
  p(t)\)/ \eps}, 
\end{equation}  
where $(q(t), p(t))\in C^3\(\R;\R^{2d}\)$ satisfies Hamilton's
equation of motion 
\begin{equation}\label{eq:classflow}
\left\{  
\begin{aligned}
  &\dot q(t)=p(t),\quad q(0)=q_0,\\
&\dot p(t)= -\nabla_xV\(q(t)\),\quad p(0)=p_0,
     \end{aligned}
\right.
 \end{equation}
and the purely time-dependent function $S(t)$ denotes the classical action (see \S\ref{sec:main} below). 
The right hand side of \eqref{eq:approx} corresponds to a wave
function which is equally localized 
in space and in frequency (at scale $\sqrt\eps$), so the uncertainty
principle is optimized. In other words, 
the three quantities
\begin{equation*}
  \|\psi^\eps(t)\|_{L^2(\R^d)},\quad \left\|\(\sqrt\eps
    \nabla-i\frac{p(t)}{\sqrt\eps}\) 
  \psi^\eps(t)\right\|_{L^2(\R^d)},\quad \text{and }
\quad \left\|\frac{x-q(t)}{\sqrt\eps}
  \psi^\eps(t)\right\|_{L^2(\R^d)}
\end{equation*}
have the same order of magnitude, $\O(1)$, as $\eps\to 0$. 
The basic idea for this type of asymptotic method can be found in the
classical works of \cite{Hag80, Lit86} (see also \cite{BiRo01, Pau97}
for a broader introduction). It has been developed further   
in, e.g., \cite{CoRo97,CoRo06, Rob10, Rou-p, SwRo09} and in
addition also proved to be applicable in the case of nonlinear
Schr\"odinger equations \cite{CaFe11} (a situation in which the use of 
Wigner measures of space-adiabatic perturbation theory fails). Asymptotic results 
based on such semiclassical wave packets also have the advantage of giving a 
rather clear connection between quantum mechanics and classical
particle dynamics  
and are thus frequently used in numerical simulations (see
e.g. \cite{FaGrLi09}). 

Ehrenfest time is the largest time up to which the wave packet
approximation is valid, in general. Without any extra geometric
assumption, the coherent structure may be lost at some time of order $C
\ln 1/\eps$, if $C$ is too large. See
e.g. \cite{BoRo02,DeRo03,Pau09,Rob02,Rob07} and references 
therein. 

Interestingly enough, though, it seems that so far this method has not been
extended to include also highly oscillatory periodic potentials
$V_\Gamma\left(\frac{x}{\eps}\right)$, and it will be the main task of
this work to do so. 
To this end, it will be necessary to understand the influence of
$V_\Gamma\left(\frac{x}{\eps}\right)$ on the dispersive properties of
the solution $\psi^\eps(t,x)$. In particular, having in mind the
results quoted above,  
one expects that in this case the usual kinetic energy of a particle
$E= \frac{1}{2} |k|^2$ has to be replaced by $E_m(k)$, i.e. the energy
of the $m$-th Bloch band associated to $V_\Gamma$.  
In physics this is known under the name \emph{Peierls substitution}. 
We shall show that under the additional influence of a slowly varying potential $V(x)$, 
this procedure is in fact asymptotically correct (i.e. for $\eps \ll
1$) up to Ehrenfest time, provided the initial data $\psi^\eps_0$ is
sufficiently concentrated around $(q_0,p_0)\in \R^{2d}$.

\begin{remark} Indeed, we could also allow for time-dependent external
  potentials $V(t,x)\in \R$ measurable in time, smooth in $x$, and
  satisfying  
  \begin{equation*}
  \d_x^\gamma V\in
 L^\infty\(\R_t\times \R_x^d\),\quad |\gamma|= 2,3.  
  \end{equation*}
Under this assumptions, it is straightforward to adapt the analysis
given below. For the sake of notation, we shall not do so here, but
rather leave the details to the reader. 
\end{remark}

\subsection{Bloch and semiclassical wave packets} \label{sec:bloch}

In order to state our result more precisely, we first recall some
well-known results on the spectral theory for periodic Schr\"odinger
operators, cf. \cite{ReSi78, Wil78}: 
\[
H_{\rm per}:= -\frac{1}{2}\Delta_y + V_\Gamma(y).
\]
Denote by $Y\subset \Gamma$ the (centered) fundamental domain of the
lattice $\Gamma$, equipped with periodic boundary conditions, i.e. $Y\simeq \mathbb T^d$. Similarly, we denote by $Y^*\simeq \mathbb T^d$ the fundamental domain
of the corresponding dual lattice. The latter is usually referred to as \emph{the Brillouin
  zone}.  
Bloch--Floquet theory asserts that $H_{\rm per}$ admits a
fiber-decomposition 
\[
H_{\rm per} = \frac{1}{|Y^*|} \int_{Y^*} H_\Gamma(k) \, dk,
\]
where for $k\in Y^*$, we denote 
\begin{equation*}
H_{\Gamma}(k)=  \frac{1}{2} \, \left(- i \nabla_y + k \right)^2+
V_{\Gamma}\left (y\right) .
\end{equation*}
It therefore suffices to consider the following spectral problem on $Y$:  
\begin{equation}\label{eq:bloch}
H_{\Gamma}(k) \chi_m (\cdot,k) = \, E_m (k)\chi_m (\cdot, k), \quad  k
  \in Y^*,\ m \in \N, 
\end{equation}
where $E_m(k)\in \R$ and $\chi_m(y,k)$, respectively, denote an
eigenvalue/eigenvector pair of $H_{\Gamma}(k)$, parametrized by
$k\in Y^*$, the so-called {\it crystal momentum.}  
These eigenvalues can be ordered increasingly, such that
\[
E_1(k)\le\ldots\le E_m(k)\le E_{m+1}(k)\le \dots,
\]
where each eigenvalue is repeated according to its multiplicity (which
is known to be finite).  
This implies that
\begin{equation*}
\text{spec} \, H_{\rm per} = \bigcup_{m \in \N} \{
E_m(k) \ ;\ k \in Y^*\} \subset \R,  
\end{equation*}
where $\{E_m (k); \, k \in Y^*\}$ is
called the $m$-th \emph{energy band} (or Bloch band).  The associated
eigenfunctions $\chi_m(y, k)$ are $\Gamma^*$--periodic w.r.t. $k$ and form a
complete orthonormal basis of $L^2(Y)$. Moreover, the functions $\chi_m(\cdot, k) \in H^2(Y)$ are 
known to be real-analytic with respect to $k$ on $Y^*\setminus \Omega$, where 
$\Omega$ is a set of Lebesgue measure zero (the set of band crossings).

Next, we consider for some $m\in \N$ the corresponding
\emph{semi-classical band Hamiltonian}, obtained by Peierl's
substitution, i.e. 
\begin{equation*}
  h_m^{\rm sc}(k,x)= E_m(k) +V(x),\quad (k,x)\in Y^*\times\R^d,
\end{equation*}
and denote the semiclassical phase space trajectories associated to
$h_m^{\rm sc}$ by 
\begin{equation}\label{eq:traj}
\left\{  
\begin{aligned}
  &\dot q(t)=\nabla_k E_m\(p(t)\),\quad q(0)=q_0,\\
&\dot p(t)= -\nabla_xV\(q(t)\),\quad p(0)=p_0. 
     \end{aligned}
\right.
 \end{equation} 
 This system is the analog of \eqref{eq:classflow} in the presence of
 an additional periodic potential. 
\begin{example}[No external potential]\label{ex:noV}
In the case $V(x)=0$, we simply have
\begin{equation}\label{eq:shift}
 p(t) = p_0, \quad q(t) = q_0 + t \nabla_k E_m(p_0) ,
 \end{equation}
 that is, a shift with constant speed $\omega = \nabla E_m(p_0)$.
 \end{example}

In order to make sure that the system  \eqref{eq:traj} is
well-defined, we shall from now on impose the following condition on
$E_m(k)$. 
 
\begin{hyp}\label{hyp:E}
 We assume that $E_m(p(t))$ is a simple eigenvalue, uniformly for all
 $t\in \R$, i.e. there exists a $\delta>0$, such that 
 \[
|E_m(p(t))- E_n(k)| \ge \delta, \quad \forall n\not=m, t\in \R, k\in Y^*.
 \]
 \end{hyp}
It is known that if $E_m(k)$ is simple, it is infinitely
differentiable and thus the right hand side of \eqref{eq:traj} is well
defined. Under Assumption~\ref{eq:classflow}, we consequently obtain a
smooth semi-classical flow $(q_0, p_0)\mapsto (q(t), p(t))$, for all
$t\in \R$.  
In addition, one can choose $\chi_m(y,k)$ to be $\Gamma$--periodic
with respect to $y$ and  
such that $(y,t)\mapsto \chi_m(y , p(t))$ is bounded together with all
its derivatives. 

\begin{example}
By compactness of $Y^*$, Assumption \ref{hyp:E} is satisfied in either
of the following two cases: 
\begin{itemize}
\item[(i)] If $E_m(k)$ is a simple eigenvalue for all $k \in Y^*$. In
  particular, in $d=1$ it is known that every $E_m(k)$ is simple,
  except possibly at $k=0$ or at the edge of the Brillouin zone. 
\item[(ii)] If $V(x)=0$ and $E_m(k)$ is simple in a neighborhood of
  $k=p_0$ (which is sufficient in view of Example \ref{ex:noV}). 
\end{itemize}
\end{example}

\subsection{Main result}
\label{sec:main}
 With the above definitions at hand, we are now able to state our main
 mathematical result. To this end, we first define a
 \emph{semiclassical wave packet} in the $m$-th Bloch band (satisfying
 Assumption \ref{hyp:E}) by 
\begin{equation}
  \label{eq:semibl}
  \varphi^\eps(t,x)=\eps^{-d/4} u 
\left(t,\frac{x-q(t)}{\sqrt\eps}\right) \chi_m\left(\frac{x}{\eps},
  p(t) \right) e^{i \phi_m(t,x) / \eps}
\end{equation}
with  $q(t), p(t)$ given by system \eqref{eq:traj} and $u(t,z)\in
\C$, a smooth slowly varying envelope which will determined by an
envelope equation yet to be derived (see below).  
In addition, the $\eps$-oscillatory phase is 
\begin{equation}\label{eq:phi}
\phi_m(t,x) = S_m(t)+p(t)\cdot (x-q(t)),
\end{equation}
where $S_m(t)\in \R$ is the (purely time-dependent) semi-classical action
\begin{equation}\label{eq:action}
S_m(t) = \int_0^t  L_m(p(s),q(s))  \, ds,
\end{equation}
with $L_m$ denoting the Lagrangian associated to the effective
Hamiltonian $h_m^{\rm sc}$, i.e. 
\begin{equation}\label{eq:lagrange}
  \begin{aligned}
    L_m(p(s),q(s)) &= p(s) \cdot \dot q(s)  - h_m^{\rm
  sc}\(p(s),q(s)\)\\
&= p(s) \cdot \nabla E_m(p(s)) -
h_m^{\rm sc}\(p(s),q(s)\) ,
\end{aligned}
\end{equation}
in view of  \eqref{eq:traj}.
\begin{remark} Note that this is nothing but the Legendre
  transform of the effective Hamiltonian $h_m^{\rm sc}$. As in
  classical mechanics, one associates to a given Hamiltonian $H(p,q)$
  a Lagrangian via $L(p,q) =   p \cdot \dot q- H(p,q)$.
\end{remark}

The function $\varphi^\eps $ given by \eqref{eq:semibl} generalizes
the usual class of semiclassical wave packets considered in
e.g. \cite{Hag80, Lit86}. Note that in contrast to  
two-scale WKB approximation considered in  \cite{BeLiPa, GuRaTr88,
  Spohn96}, it involves an {\it additional scale} of the order
$1/\sqrt{\eps}$, the 
scale of concentration of the amplitude $u$. 
In addition, \eqref{eq:semibl} does {\it not} suffer from the 
appearance of caustics.  
Nevertheless, in comparison to the highly oscillatory Bloch function
$\chi_m$, the amplitude is still slowly varying and thus we can expect
an adiabatic decoupling between the slow and fast scales  
to hold on (long) macroscopic time-scales. 
Indeed, we shall prove the following result:

\begin{theorem}\label{th:main} 
Let $V_\Gamma $ be smooth and $V$ satisfy Assumptions \ref{hyp:V}. In
addition, let Assumption~\ref{hyp:E} hold  
and the initial data be given by 
\[
\psi^\eps_0(x) = \eps^{-d/4} u_0 
\left(\frac{x-q_0}{\sqrt\eps}\right) \chi_m\left(\frac{x}{\eps},
  p_0 \right) e^{i p_0 \cdot (x-q_0)/ \eps}, 
\]
with $q_0,p_0 \in \R^d$ and some given profile $u_0 \in \mathcal
S(\R^d)$. Then there 
exists $C>0$ such that the solution of \eqref{eq:NLS0} can be approximated by
\[
\| \psi^\eps (t) - \varphi^\eps(t) \|_{L^2(\R^d)} \le C \sqrt{\eps} e^{C t}.
\]
Here, $\varphi^\eps$ is given by \eqref{eq:semibl} with
\[
u(t,z) = v(t,z) \exp\(\int_0^t  \beta(\tau) d\tau\),
\]
where $\beta(t)\in i \R$ is the so-called Berry phase term
\[
\beta(t): =  \< \chi_m( p(t)), \nabla_k \chi_m(p(t)) \>_{L^2(Y)} \cdot
\nabla V(q(t)), 
\]
and $v \in
C(\R;\Sch(\R^d))$ satisfies the following homogenized Schr\"odinger
equation 
\begin{equation}\label{eq:v}
  i\d_t v + \frac{1}{2}{\rm div}_z\(\nabla^2_k E_m\(p(t)\)\cdot
  \nabla_z\) v = \frac{1}{2} \< z, \nabla_x^2V\(q(t)\) z \> v, \quad
  v_{\mid t=0} = u_{0}.
\end{equation} 
In particular there exists $C_0>0$ so that
\[
\sup_{0\le t\le C_0 \ln \frac{1}{\eps}}\| \psi^\eps (t) -
\varphi^\eps(t) \|_{L^2(\R^d)} \Tend \eps 0 0. 
\] 
\end{theorem}

\begin{remark} In fact it is possible to prove the same result under
  less restrictive regularity assumptions on $u_0$ and $V_\Gamma$.  
Indeed, Proposition \ref{prop:stab} shows that it is sufficient to
require that $u_0$ belongs to a certain weighted Sobolev space.  
Concerning the periodic potential, it is possible to lower the
regularity considerably, depending on the dimension. For example, in
$d=3$ it is  
sufficient to 
assume $V_\Gamma$ to be infinitesimally bounded with respect to 
$-\Delta$. This implies $\chi_m(\cdot, k) \in H^2(\R^3)
\hookrightarrow L^\infty(\R^3)$, which, together with several  
density arguments (to be invoked at different stages of the formal
expansion), is enough to justify the analysis given below. 
\end{remark}

Theorem \ref{th:main} provides an approximate description of the solution to
\eqref{eq:NLS0} up to Ehrenfest time and can be seen as the analog of
the results given in   
\cite{Hag80, Lit86, CoRo97,CoRo06, HaJo01, Rob10, Rou-p, SwRo09} where
the case of slowly varying potentials $V(x)$ is considered.  The proof
does not rely on the use of  pseudo-differential calculus or space space-adiabatic perturbation 
theory and can thus be considered to be considerably simpler from a mathematical point of view.
In fact, our approach is similar to the one given in 
\cite{HaJo01}, which derives an analogous result for the so-called 
Born-Oppenheimer approximation of molecular dynamics. Note however, that we 
allow for more general initial amplitudes, not necessarily Gaussian. 
Indeed, in the special case where the initial envelope $u_0$ is a Gaussian,
then its evolution $u$ remains Gaussian, and can be completely
characterized; see \S\ref{sec:gaussian}. 
Also note that in contrast to the closely related method of Gaussian beams
presented in, e.g., \cite{DiGuRa06}, we do not need to include
complex-valued phases and in addition, obtain an  
approximation valid for longer times.

The Berry phase term is an example for so-called \emph{geometric
  phases} in quantum mechanics. It is a well known feature of  
semiclassical approximation in periodic potentials, see, e.g.,
\cite{PaSpTe03} for more details and a geometric interpretation. 
The homogenized Schr\"odinger equation features a rather
unusual dispersive behavior described  
by a time-dependent effective mass tensor $M(t)=\nabla^2_k E_m\(p(t)\)$, i.e. 
the Hessian of $E_m(k)$ evaluated at $k=p(t)$. 
To our knowledge, Theorem \ref{th:main} is the first result in
which a Schr\"odinger type equation with \emph{time-dependent} effective
mass has been rigorously derived (see also the discussion in Remark
\ref{rem:effective}).

\begin{remark} Let us also mention that the same class of initial data
  has been considered in \cite{AlPo06} for a Schr\"odinger equation
  with \emph{locally} periodic potential $V_\Gamma(x, y)$  
and corresponding $x$-dependent Bloch bands $E_m(k;x)$. 
In this work, the authors derive a homogenized Schr\"odinger equation, provided 
that $\psi_0^\eps$ is concentrated around a \emph{stationary point}
point $x_0, p_0$ of the semiclassical phase space flow, i.e.  
\[
\nabla_k E_m(p_0; q_0) = \nabla_x E_m(p_0 ;q_0) = 0.
\]
This implies $q(t)=q_0$ and $p(t)=p_0$, for all $t\in \R$, yielding (at least asymptotically) a 
localization of the wave function.  
We observe the same phenomenon in our case 
under the condition $V(x)= 0$ and $\nabla_k E_m(k)=0$
(see Example \ref{ex:noV}). 
\end{remark}

This work is now organized as follows: In the next section, we shall
formally derive an approximate solution to \eqref{eq:NLS0} by means of
a (formal) multi-scale expansion.  
This expansion yields a system of three linear equations, which we
shall solve in Section \ref{sec:alg}. 
In particular, we shall obtain from it the homogenized Schr\"odinger
equation. The corresponding Cauchy problem is then analyzed in Section
\ref{sec:profile}, where we also include  
a brief discussion on the particularly important case of Gaussian
profiles (yielding a direct connection to \cite{Hag80}). A rigorous  
stability result for our approximation, up to Ehrenfest time, is then
given in Section \ref{sec:stab}. 

\begin{remark} We expect that our results can be generalized to the case
of (weakly) nonlinear Schr\"odinger equations (as considered in \cite{ CaFe11, CaMaSp}). 
This will be the aim of a future work. 
\end{remark}

\section{Formal derivation of an approximate
  solution} \label{sec:formal}  

\subsection{Reduction through exact computations}
\label{sec:exact}

We seek the solution $\psi^\eps$ of \eqref{eq:NLS0} in the following form
\begin{equation}
  \label{eq:change}
  \psi^\eps(t,x)= \eps^{-d/4}\,
  \U^\eps\(t,\frac{x-q(t)}{\sqrt\eps},\frac{x}{\eps}\) e^{i \phi_m(t,x) /
    \eps}, 
\end{equation}
where the phase $\phi(t,x)$ is given by \eqref{eq:phi}, the function
$\U^\eps=\U^\eps(t,z,y)$ is assumed to be smooth, $\Gamma$-periodic
with 
respect to $y$, and admits an asymptotic expansion
\begin{equation}\label{eq:DA}
  \U^\eps(t,z,y)\Eq \eps 0 \sum_{j\ge 0}\eps^{j/2}U_j(t,z,y). 
\end{equation}
Note that due to the inclusion of the factor $\eps^{-d/4}$ the
$L^2(\R^d)$ norm of the right hand side of \eqref{eq:change} is in
fact uniformly bounded with respect to $\eps$,  
whereas the $L^\infty(\R^d)$ norm in general will grow as $\eps \to
0$. The asymptotic expansion \ref{eq:DA} therefore has to be
understood in the $L^2$ sense. 

Taking into account that in view of \eqref{eq:phi},
$\nabla_x\phi_m(t,x)=p(t)$, we compute:  
\begin{align*}
  \eps^{d/4}e^{-i\phi_m/\eps}i\eps \d_t \psi^\eps &= i\eps \d_t \U^\eps
   -i\sqrt\eps \dot q \cdot \nabla_z \U^\eps -\d_t \phi_m \U^\eps ,\\
\eps^{d/4}e^{-i\phi_m/\eps}\eps^2 \Delta \psi^\eps &= \eps \Delta_z
\U^\eps  +\Delta_y \U^\eps + 2\sqrt\eps \(\nabla_y\cdot
\nabla_z\)\U^\eps-|p|^2 \U^\eps \\
&\quad +2i \sqrt\eps p\cdot \nabla_z
\U^\eps + 2i p\cdot \nabla_y \U^\eps,
\end{align*}
where in all of the above expressions, the various functions have to
be understood to be evaluated as follows: 
\begin{equation*}
  \psi^\eps=\psi^\eps(t,x)\quad ;\quad \U^\eps=\U^\eps
  \(t,\frac{x-q(t)}{\sqrt\eps},\frac{x}{\eps}\). 
\end{equation*}
Thus, ordering equal powers of $\eps$ in equation \eqref{eq:phi} we find that
\begin{equation*}
  \eps^{d/4}e^{-i\phi_m/\eps}\(i\eps \d_t
  \psi^\eps+\frac{\eps^2}{2}\Delta \psi^\eps - V_\Gamma
  \(\frac{x}{\eps}\) \psi^\eps 
  -V(x)\psi^\eps\) = b_0^\eps +\sqrt\eps b_1^\eps +\eps b_2^\eps, 
\end{equation*}
with
\begin{align*}
b_0^\eps&= -\d_t \phi_m \U^\eps  +\frac{1}{2} \Delta_y \U^\eps
-\frac{1}{2}\lvert p\rvert^2\U^\eps 
+ip\cdot \nabla_y \U^\eps
-V_\Gamma(y) \U^\eps
-V(q)\U^\eps
,\\
b_1^\eps&= -i \dot q \cdot \nabla_z \U^\eps
+\(\nabla_y\cdot\nabla_z\)\U^\eps +i p \cdot \nabla_z \U^\eps,\\
b_2^\eps&= i\d_t \U^\eps +\frac{1}{2}\Delta_z \U^\eps.
\end{align*}
So far, we have neither used the fact that $q(t)$, $p(t)$ are given by the
Hamiltonian flow \eqref{eq:traj}, nor the explicit dependence of
$\phi_m$ on time. Using these properties, allows us to rewrite
$b_0^\eps, b_1^\eps, b_2^\eps$ as follows: 
\begin{align*}
  b_0^\eps&=  \(h_m^{\rm sc}(p(t),q(t))+\nabla V\(q(t)\)\cdot \(x-q(t)\)\)\U^\eps
  -  H_\Gamma\(p(t)\) \U^\eps -V (q(t))\U^\eps ,\\
b_1^\eps &= i\(p(t) -\nabla_k E_m\(p(t)\)\)\cdot \nabla_z \U^\eps + 
\(\nabla_y\cdot\nabla_z\)\U^\eps,\\
b_2^\eps &= i\d_t \U^\eps +\frac{1}{2}\Delta_z \U^\eps.
\end{align*}
Now, recall that in  the above lines, $\U^\eps$ is evaluated at the
shifted spatial variable $z=(x-q(t))/\sqrt\eps.$ Taking 
this into account, we notice that the above hierarchy has to be
modified, and we find:
\begin{align*}
  b_0^\eps &=  h_m^{\rm sc}(p(t),q(t))\, \U^\eps
  -  H_\Gamma\(p(t)\) \U^\eps -V\(q(t)+z\sqrt\eps \) \U^\eps ,\\
b_1^\eps&= i\(p(t) -\nabla_k E_m\(p(t)\)\)\cdot \nabla_z \U^\eps + 
\(\nabla_y\cdot\nabla_z\)\U^\eps+(\nabla V\(q(t)\)\cdot z)\, \U^\eps,\\
b_2^\eps &= i\d_t \U^\eps +\frac{1}{2}\Delta_z \U^\eps.
\end{align*}
Next, we perform a Taylor expansion of $V$ around the point $q(t)$:
\begin{equation*}
  V\(q(t)+z\sqrt\eps \)= V\(q(t)\) + \sqrt\eps \nabla V\(q(t)\)\cdot z +
  \frac{\eps}{2} \< z, \nabla^2V\(q(t)\) z \>+ \O\( \eps^{3/2}\<z\>^3\),
\end{equation*}
since $V$ is at most quadratic in view of Assumption
\ref{hyp:V}. Recalling that $h_m^{\rm sc}(p,q)=E_m(p)+V(q)$, the terms 
involving $V(q)$ cancel out in $b_0^\eps$, the terms involving $\nabla
V(q)$ cancel out in $b_1^\eps$, and thus, we finally obtain:
\begin{lemma}\label{lem:form}
  Let the Assumptions \ref{hyp:V}, \ref{hyp:E} hold and $\psi^\eps$ be
  related to $\U^\eps$ through 
  \eqref{eq:change}. Then it holds
  \begin{align*}
    i\eps \d_t
  \psi^\eps& +\frac{\eps^2}{2}\Delta \psi^\eps - V_\Gamma
  \(\frac{x}{\eps}\)\psi^\eps 
  -V(x)\psi^\eps =\\
& \frac{e^{i\phi_m/\eps}}{\eps^{d/4}}\( b_0^\eps +\sqrt\eps
  b_1^\eps  +\eps b_2^\eps +\eps^{3/2}r^\eps\)(t,z,y)\Big|_{(z,y)=
    \(\frac{x-q(t)}{\sqrt\eps},\frac{x}{\eps}\) },
  \end{align*}
with 
\begin{align*}
  b_0^\eps & = \(E_m\(p(t)\) -H_\Gamma \(p(t)\)\)\U^\eps,\\
b_1^\eps  & = i\(p(t) -\nabla_k E_m\(p(t)\)\)\cdot \nabla_z \U^\eps + 
\(\nabla_y\cdot\nabla_z\)\U^\eps,\\
b_2^\eps&= i\d_t \U^\eps +\frac{1}{2}\Delta_z \U^\eps - \frac{1}{2}
 \< z, \nabla^2V\(q(t)\) z \> \, \U^\eps,
 \end{align*}
and a remainder $r^\eps(t,z,y)$ satisfying
\begin{align*}
|r^\eps(t,z,y)|&\le C \<z\>^3 |\U^\eps(t,z,y)|,\quad \forall
(t,z,y)\in \R\times\R^d\times Y, 
\end{align*}
where the constant $C>0$ is independent of $t,z,y$ and $\eps$. 
\end{lemma}

\subsection{Introducing the approximate solution}\label{sec:app}

We now expand $\U^\eps$ in powers of $\eps$, according to
\eqref{eq:DA}. To this end, we introduce the following
(time-dependent) linear operators 
\begin{align*}
  L_0& = E_m\(p(t)\) -H_\Gamma \(p(t)\),\\
 L_1 &= i\(p(t) -\nabla_k E_m\(p(t)\)\)\cdot \nabla_z + 
\nabla_y\cdot\nabla_z,\\
L_2&= i\d_t  +\frac{1}{2}\Delta_z  - \frac{1}{2}
 \< z, \nabla^2V\(q(t)\) z \>. 
\end{align*}
In order to solve \eqref{eq:NLS0} up to a sufficiently small error
term (in $L^2$), 
we need to cancel the first three terms in our asymptotic
expansion. This yields, the following system of equations 
\be
\label{eq:systU}
\left \{
\begin{aligned}
  &L_0 U_0=0,\\
&L_0 U_1+L_1U_0=0,\\
&L_0 U_2+L_1 U_1+L_2 U_0=0.
\end{aligned}
\right.
\ee
Assuming for the moment that we can do so, this means that we (formally) solve \eqref{eq:NLS0} up to errors of order $\eps^{3/2}$ (in $L^2$),
which is expected to generate a small perturbation of the exact
solution (in view of the $\eps$ in front of the time derivative of
$\psi^\eps$ in \eqref{eq:NLS0}). 

We consequently define the approximate solution 
\begin{equation}
  \label{eq:app}
  \psi_{\rm app}^\eps(t,x):= \eps^{-d/4}
  \(U_0 + \sqrt{\eps} U_1 + \eps U_2 \)
  \(t,\frac{x-q(t)}{\sqrt\eps},\frac{x}{\eps}\) e^{i \phi_m(t,x) / 
    \eps}.
\end{equation}
In view of Lemma \ref{lem:form}, we thus have the following result (provided we can solve the system \eqref{eq:systU} in a unique way):
\begin{lemma}\label{lem:app}
  Let $\psi_{\rm app}^\eps$ given by \eqref{eq:app}, where $U_0, U_1,
  U_2$ solve \eqref{eq:systU}. Then 
  \begin{align*}
    \(i\eps \d_t
   +\frac{\eps^2}{2}\Delta  -
  V_\Gamma 
  -V \)\psi_{\rm app}^\eps=
 \frac{e^{i\phi_m/\eps}}{\eps^{d/4}}
 \eps^{3/2}\(r^\eps+r_1^\eps+r_2^\eps\) (t,z,y)\Big|_{(z,y)= 
    \(\frac{x-q(t)}{\sqrt\eps},\frac{x}{\eps}\) },
  \end{align*}
where the remainder terms $r^\eps_1, r^\eps_2$ are given by
\begin{align*}
r_1^\eps(t,z,y)=L_2U_1(t,z,y),\quad r_2^\eps(t,z,y)=L_1U_2(t,z,y),
\end{align*}
and $r^\eps$ satisfies 
\begin{align*}
|r^\eps(t,z,y)|\le C \<z\>^3 \left\lvert \(U_0+\sqrt \eps U_1+\eps
  U_2\)(t,z,y)\right\rvert,\quad \forall 
(t,z,y)\in \R\times\R^d\times Y, 
\end{align*}
where the constant $C>0$ is independent of $t,z,y$ and $\eps$. 
\end{lemma}

\section{Derivation of the homogenized equation}
\label{sec:alg}

\subsection{Some useful algebraic identities} Given the form of $L_0$, the equation $L_0U_0=0$ implies
\begin{equation}\label{eq:U0}
  U_0(t,z,y)= u(t,z)\chi_m\(y,p(t)\). 
\end{equation}
Before studying the other two equations, we shall recall some algebraic
formulas related to the eigenvalues and eigenvectors of
$H_\Gamma$. First, in view of the identity \eqref{eq:bloch}, we have
\begin{equation}\label{eq:dbloch}
  \nabla_k\(H_\Gamma -E_m\) \chi_m + \(H_\Gamma -E_m\) \nabla_k\chi_m=0.
\end{equation}
Taking the scalar product in $L^2(Y)$ with $\chi_m$, we infer
\begin{align*}
  \nabla_k E_m = \<\chi_m, \nabla_k H_\Gamma \chi_m\>_{L^2(Y)}
  +\<\chi_m,\(H_\Gamma -E_m\) \nabla_k\chi_m\>_{L^2(Y)}.
\end{align*}
Since $H_\Gamma $
is self-adjoint, the last term is zero, thanks to \eqref{eq:bloch}. We
infer
\begin{equation}
  \label{eq:dkE}
  \nabla_k E_m(k) = \<\chi_m,\(-i\nabla_y+k\)\chi_m\>_{L^2(Y)}.
\end{equation}
Differentiating \eqref{eq:dbloch} again, we have, for all $j,\ell \in
\{1,\dots,d\}$: 
\begin{align*}
  \d_{k_j k_\ell}^2 \(H_\Gamma -E_m\) \chi_m &+ \d_{k_j}\(H_\Gamma -E_m\)
  \d_{k_\ell}\chi_m + \d_{k_\ell}\(H_\Gamma -E_m\) \d_{k_j}\chi_m\\
&+\(H_\Gamma -E_m\)
  \d_{k_j k_\ell}^2\chi_m=0. 
\end{align*}
Taking the scalar product with $\chi_m$, we have:
\begin{equation}
  \label{eq:d2kE}
  \begin{aligned}
      \d_{k_j k_\ell}^2 E_m(k) &= \delta_{j\ell} +\< \(-i
  \d_{y_j}+k_j\)\d_{k_\ell}\chi_m +\(-i
  \d_{y_\ell}+k_\ell\)\d_{k_j}\chi_m  ,\chi_m\>_{L^2(Y)}\\
&\quad -\< \d_{k_\ell}E_m\d_{k_j}\chi_m +\d_{k_j}E_m\d_{k_\ell}\chi_m
,\chi_m\>_{L^2(Y)}. 
  \end{aligned}
\end{equation}

\subsection{Higher order solvability conditions} By Fredholm's alternative, a necessary {\it and} sufficient condition to
solve the equation $L_0U_1+L_1U_0=0$, is that $L_1U_0$ is orthogonal to
$\ker L_0$, 
that is:
\begin{equation}\label{eq:orth1}
  \<\chi_m,L_1 U_0\>_{L^2(Y)}=0. 
\end{equation}
Given the expression of $L_1$ and the formula \eqref{eq:U0}, we
compute
\begin{equation*}
  L_1U_0 = i\(p(t)-\nabla_kE_m\(p(t)\)\)\cdot \nabla_z u(t,z)
  \chi_m\(y,p(t)\) + \nabla_y\chi_m\(y,p(t)\)\cdot \nabla_z u(t,z). 
\end{equation*}
In view of \eqref{eq:dkE}, we infer that \eqref{eq:orth1} is
automatically fulfilled. We thus obtain
\begin{equation*}
  U_1(t,z,y) = u_1(t,z)\chi_m\(y,p(t)\) + u_1^\perp(t,z,y),
\end{equation*}
where $u_1^\perp$, the part of $U_1$ which is orthogonal to $\ker
L_0$, is obtained by inverting an elliptic equation: 
\begin{equation*}
  u_1^\perp = -L_0^{-1}L_1 U_0.
\end{equation*}
Note that the formula for $L_1U_0$ can also be written as
\begin{equation*}
  L_1U_0 =-i \nabla_k \(E_m\(p(t)\)-H_\Gamma\(p(t)\)\)
  \chi_m\(y,p(t)\) \cdot \nabla_z u(t,z),
\end{equation*}
thus taking into account \eqref{eq:dbloch}, we simply have:
\begin{equation*}
  u_1^\perp (t,z,y) = -i \nabla_k \chi_m\(y,p(t)\)\cdot \nabla_z
  u(t,z). 
\end{equation*}
At this stage, we shall, for simplicity choose $u_1=0$, in which case
$U_1$ becomes simply a function of $u$:
\begin{equation}\label{eq:U1}
  U_1(t,z,y) = -i \nabla_k \chi_m\(y,p(t)\)\cdot \nabla_z
  u(t,z). 
\end{equation}

As a next step in the formal analysis, we must solve
\begin{equation*}
  L_0 U_2+L_1 U_1+L_2 U_0=0.
\end{equation*}
By the same argument as before, we require
\begin{equation}\label{eq:orth2}
  \<\chi_m,L_1 U_1+L_2 U_0\>_{L^2(Y)}=0. 
\end{equation}
With the expression \eqref{eq:U1}, we compute
\begin{equation*}
  L_1U_1 = \sum_{j,\ell=1}^d \( \(p(t)-\nabla_k E_m\(p(t)\)\)_j
  \d_{k_\ell}\chi_m\(y,p(t)\) -i \d_{k_j
    k_\ell}^2\chi_m\(y,p(t)\)\)\d_{z_j z_\ell}^2 u,
\end{equation*}
and we also have
\begin{align*}
  L_2U_0 &= \(\( i\d_t +\frac{1}{2}\Delta_z
  -\frac{1}{2} \< z, \nabla^2V\(q(t)\) z \> \)u\)\chi_m\(y,p(t)\)\\
&\quad +
  iu\, \dot p(t)\cdot \nabla_k \chi_m\(y,p(t)\). 
\end{align*}
Recalling that $\dot p(t)= - \nabla V(q(t))$, we find:
  \begin{align*}
  & \<\chi_m,L_1 U_1+L_2 U_0\>_{L^2(Y)}= \\
  & = \( i\d_t +\frac{1}{2}\Delta_z
  -\frac{1}{2} \< z, \nabla^2V\(q(t)\) z \>\)u
 -i \nabla V(q(t)) \cdot  \< \chi_m, \nabla_k  \chi_m  \>_{L^2(Y)}  u  \\
 & -\sum_{j,\ell} \< \chi_m,\d_{k_j} E_m\(p(t)\)
  \d_{k_\ell}\chi_m +i \d_{k_j
    k_\ell}^2\chi_m\>_{L^2(Y)}\d_{z_j z_\ell}^2 u
  \end{align*}
By making the last sum symmetric with respect to $j$ and $\ell$, and
using \eqref{eq:d2kE}, we finally obtain the homogenized Schr\"odinger
equation with time-dependent effective mass: 
\begin{equation}
  \label{eq:u}
  i\d_t u + \frac{1}{2}{\rm div}_z\(\nabla^2_k E_m\(p(t)\)\cdot
  \nabla_z\) u = \frac{1}{2} \< z, \nabla^2V\(q(t)\) z \> u + i \beta(t) u, 
\end{equation}
where we recall that 
\[
\beta(t)= \nabla V(q(t))\cdot \< \chi_m( p(t)), \nabla_k \chi_m(p(t))
\>_{L^2(Y)} , 
\]
the so-called Berry phase term. From $\| \chi_m \|_{L^2(Y)} =1$, we infer that 
$$\text{Re}  \< \chi_m, \nabla_k \chi_m \>_{L^2(Y)} = 0.$$ In other words, 
$ \< \chi_m, \nabla_k \chi_m \>_{L^2(Y)} \in i \R$ and thus $i
\beta(t)\in \R$, acts like a purely time-dependent, real-valued,
potential. Thus, invoking the unitary change of variable 
\[
v(t,z) = u(t,z) e^{ - \int_0^t \beta(s) ds}
\]
implies that $v(t,z)$ solves \eqref{eq:v}.  
Equation~\eqref{eq:u} models a quantum mechanical time-dependent
harmonic oscillator, in which the time 
dependence is present both in the differential operator, and in the
potential. 
\begin{remark} \label{rem:effective}
In the case where $V(x)=0$, we have  $\beta(t)=0$ and $p(t)=p_0$
(in view of  Example~\ref{ex:noV}). In this case $v(t,z) = u(t,z)$ and  
Equation~\eqref{eq:u} simplifies to an equation with a 
{\it time-independent} effective mass tensor:
\begin{equation*}
  i\d_t u + \frac{1}{2}{\rm div}_z\(\nabla^2_k E_m\(p_0\)\cdot
  \nabla_z\) u =0.
\end{equation*}
This equation has been derived in \cite{AlPi05} (see also \cite{Sp,
  HoWe10, IlWe10} for similar results).  
Note, however, that in the quoted works the scaling of the original
equation \eqref{eq:NLS0} is different (i.e. not in semiclassical
form). 
\end{remark}
Assuming for the moment that we can solve \eqref{eq:u}, we can write
\begin{equation}\label{eq:U2}
  U_2(t,z,y)= u_2(t,z)\chi_m\(y,p(t)\) + u_2^\perp(t,z,y),
\end{equation}
where
\begin{equation*}
  u_2^\perp = -L_0^{-1}\(L_1 U_1+L_2 U_0\).
\end{equation*}
Like we did for $u_1$, we shall from now on also impose $u_2\equiv 0$ and thus $U_2 = u_2^\perp$. For the upcoming analysis, we define the following class of energy spaces
\begin{equation*}
  \Sigma^k = \left\{ f\in L^2(\R^d)\ ;\ \| f \| _{\Sigma^k}:=
    \sum_{|\alpha|+|\beta|\le  k}\left\lVert x^\alpha \d_x^\beta
      f\right\rVert_{L^2(\R^d)}<\infty; \ k \in \N \right\}. 
\end{equation*}

Having in mind \eqref{eq:U0}, \eqref{eq:U1}, \eqref{eq:U2}, and the fact that $L_0^{-1}: L^2(Y)\to H^2(Y)$, we directly obtain the following result:

\begin{proposition}\label{prop:app} Let Assumption \ref{hyp:E} hold and let $u\in C(\R; \Sigma^k)$ solve \eqref{eq:u}. Set 
\begin{align*}
U_0 (t,z,y) &= u(t,z) \chi_m \(y,p(t)\), \\
U_1 (t,z,y)&= -i \nabla_k \chi_m\(y,p(t)\)\cdot \nabla_z
  u(t,z), \\
U_2(t,z,y)&=  -L_0^{-1}\(L_1 U_1(t,z,y)+L_2 U_0(t,z,y)\).
\end{align*}
Then $U_j \in C(\R; \Sigma^{k-j}_z\times W^{\infty, \infty} (Y))$, for $j=0,1,2$ and $(U_0,U_1,U_2)$ solves \eqref{eq:systU}.
\end{proposition}

\section{The envelope equation}
\label{sec:profile}

We examine the Cauchy problem for \eqref{eq:u}, with special emphasis
on the large time control of $u$. 

\subsection{The general Cauchy problem}
\label{sec:cauchygen}
Equation~\ref{eq:u} can be seen as the quantum mechanical evolutionary problem corresponding to the following time-dependent
Hamiltonian,
\begin{equation}\label{eq:hamil}
 H(t,z,\zeta)= \frac{1}{2}\<\zeta,\nabla^2_kE_m\(p(t)\)\zeta\> +
  \frac{1}{2}\<z\nabla^2 V\(q(t)\)z\> + i \beta(t).
\end{equation}
Under Assumptions~\ref{hyp:V} and \ref{hyp:E},
this Hamiltonian is self-adjoint, smooth in time, and quadratic in
$(z,\zeta)$  (in fact, at most quadratic would be sufficient). Using the
result given in \cite[p.197]{Kit80} (see also 
\cite{KiKu81}), we directly infer the following existence result:
\begin{lemma}[From \cite{Kit80}]\label{lem:cauchy}
  For  $d\ge 1$ and ${\bf v}_0\in L^2(\R^d)$, consider the equation
  \begin{equation}\label{eq:ugen}
    i\d_t {\bf v} +\frac{1}{2}\sum_{1\le j,k\le d}a_{jk}(t)\d_{jk}^2 {\bf v} =
    \frac{1}{2}\sum_{1\le j,k\le d} b_{jk}(t)x_jx_k {\bf v} \quad ;\quad
    {\bf v} _{\mid t=0} = {\bf v}_0.
  \end{equation}
If the coefficients $a_{jk}$ and $b_{jk}$ are continuous and
real-valued, such 
that the matrices $(a_{jk})_{j,k}$ and $(b_{jk})_{j,k}$ are symmetric
for all time, then \eqref{eq:ugen} has a unique solution ${\bf v} \in
C(\R;L^2(\R^d))$. It satisfies
\begin{equation}\label{eq:est}
  \|{\bf v} (t)\|_{L^2(\R^d)} = \|{\bf v}_0\|_{L^2(\R^d)},\quad \forall t\in \R.
\end{equation}
Moreover, if ${\bf v}_0\in \Sigma^k$ for some $k\in \N$, then ${\bf v} \in
C(\R;\Sigma^k)$. 
\end{lemma}
In particular, this implies that if $u_0\in \Sigma^k$, then
\eqref{eq:v} has a unique 
solution $v\in C(\R;\Sigma^k)$. As a consequence, \eqref{eq:u} has a
unique solution $u\in C(\R;\Sigma^k)$ such that $u_{\mid t=0}=u_0$. 
\begin{remark}
  It may happen that the functions $a_{jk}$ are zero on some
  non-negligible set. In this case, \eqref{eq:ugen} ceases to be
  dispersive. Note that the standard harmonic oscillator is dispersive,
  \emph{locally in time} only, since it has eigenvalues. We shall see
  that this is not a problem in our analysis 
  though. 
\end{remark}
\subsection{Exponential control of the envelope equation}\label{sec:exp}
To prove Theorem~\ref{th:main}, we need to control the error present
in Lemma~\ref{lem:app} for large time. In general, i.e. without
extra geometric assumptions on the wave packet, exponential growth in
time must be expected:
\begin{proposition}\label{prop:growth}
  Let $u_0\in \Sigma^k$, $k\in \N$. Then the solution $u$ to
  \eqref{eq:u} satisfies $u\in
C(\R;\Sigma^k)$, and there exists $C>0$ such that  
  \begin{equation*}
    \left\lVert  u(t,\cdot)\right\rVert_{\Sigma^k}\le C
    e^{Ct}, \quad t\ge0. 
  \end{equation*}
\end{proposition}
\begin{proof}
  The result can be established by induction on $k$. The constant $C$
  must actually be expected to depend on $k$, as shown by the case of
  \begin{equation*}
    i\d_t u + \frac{1}{2}\Delta_z u = -\frac{|z|^2}{2} u\quad ;\quad
    u_{\mid t=0} = u_0.
  \end{equation*}
There, the fundamental solution is explicit (generalized Mehler
formula, see e.g. \cite{Hor95}), and we check that 
$\|u(t)\|_{\Sigma^k}$ behaves like $e^{k t}$.  

For $k=0$, the result is obvious, since in view of
Lemma~\ref{lem:cauchy}, the $L^2$-norm is conserved. The case $k=1$
illustrates the general mechanism of the proof, and we shall stick to
this case for simplicity. The key remark is that even though the
operators $z$ and $\nabla_z$ (involved in the definition of
$\Sigma^1$) do not commute with the Hamiltonian \eqref{eq:hamil}, the
commutators yield a closed system of estimates. First, multiplying
\eqref{eq:u} by $z$, we find
\begin{equation*}
  \(i\d_t-H\)zu = -\left[ H,z\right] u = \nabla^2_k
  E_m\(p(t)\)\nabla_z u.
\end{equation*}
The $L^2$ estimate \eqref{eq:est} then yields
\begin{align*}
  \left\lVert z u(t)\right\rVert_{L^2}&\le \left\lVert z
    u_0\right\rVert_{L^2}+ \int_0^t \left\lVert\nabla^2_k
  E_m\(p(s)\)\nabla_z u (s)\right\rVert_{L^2}ds \\
&\le \left\lVert z
    u_0\right\rVert_{L^2}+ C\int_0^t \left\lVert\nabla_z u
    (s)\right\rVert_{L^2}ds, 
\end{align*}
for some $C$ independent of $t$, since $\nabla^2_k
  E_m$ is bounded on $Y^*$ which is compact. Similarly,
  \begin{equation*}
    \(i\d_t-H\)\nabla_zu = -\left[ H,\nabla_z\right] u = \nabla^2
  V\(q(t)\)z u,
  \end{equation*}
and, in view of Assumption~\ref{hyp:V},
\begin{equation*}
  \left\lVert \nabla_z u(t)\right\rVert_{L^2}\le \left\lVert \nabla_z
    u_0\right\rVert_{L^2}+ C\int_0^t \left\lVert z u
    (s)\right\rVert_{L^2}ds. 
\end{equation*}
Summing over the two inequalities and using the conservation of mass,
we infer 
\begin{equation*}
   \left\lVert u(t)\right\rVert_{\Sigma^1}\le \left\lVert 
    u_0\right\rVert_{\Sigma^1}+ C\int_0^t \left\lVert  u
    (s)\right\rVert_{\Sigma^1}ds, 
\end{equation*}
and Gronwall's lemma yields the proposition in the case $k=1$. By
induction, applying $(z,\nabla_z)$ to \eqref{eq:u} $k$ times, the defects of
commutation always yield the same sort of estimate, and the
proposition follows easily.
\end{proof}

\subsection{Gaussian wave packets}
\label{sec:gaussian}

In the case where the initial datum in \eqref{eq:u} is a Gaussian, we
can compute its evolution and show that it remains Gaussian, by
following the same strategy as in
\cite{Hag80} (see also \cite{HaJo00,HaJo01}). As a matter of fact, the
order in which we have proceeded is different from the one in
\cite{Hag80}, since we have isolated the envelope equation
\eqref{eq:u} before considering special initial data. As a
consequence, we have fewer unknowns. Consider \eqref{eq:u} with
initial datum
\begin{equation}\label{eq:CIG}
  u(0,z) = \frac{1}{({\rm det} A)^{1/2}}\exp \(-\frac{1}{2} \<z,BA^{-1}z\>\),
\end{equation}
where the matrices $A$ and $B$ satisfy the following properties:
\begin{align}
&  A\text{ and } B \text{ are invertible};\label{eq:ABinv}\\
&BA^{-1} \text{ is symmetric}: BA^{-1} = M_1+iM_2,\text{ with }M_j
\text{ real symmetric};\label{eq:symm}\\
&\RE BA^{-1} \text{ is strictly positive definite};\label{eq:pos}\\
&\(\RE BA^{-1}\)^{-1}=AA^*.\label{eq:invBA}
\end{align}
\begin{proposition}
  Let $u$ solve \eqref{eq:u}, with initial datum \eqref{eq:CIG}, where
  the matrices $A$ and $B$ satisfy
  \eqref{eq:ABinv}--\eqref{eq:invBA}. Then for all time, $u(t,z)$ is
  given by
  \begin{equation}\label{eq:uG}
    u(t,z)= \frac{1}{({\rm det} A(t))^{1/2}}\exp \(-\frac{1}{2}\<z,B(t)A(t)^{-1}z\>\),
  \end{equation}
where the matrices $A(t)$ and $B(t)$ evolve according to  the
differential equations 
\begin{equation}
  \label{eq:AB}
\left\{
\begin{aligned}
 \dot A(t)& = i\nabla^2_k E_m\(p(t)\)B(t)\quad ;\quad A(0)=A,\\
\dot B(t)& = i\nabla^2_x V\(q(t)\)A(t)\quad ;\quad B(0)=B.
\end{aligned}
\right.  
\end{equation}
In addition, for all time $t\in \R$, $A(t)$ and $B(t)$ satisfy
\eqref{eq:ABinv}--\eqref{eq:invBA}. 
\end{proposition}
\begin{proof}
  The argument is the same as in \cite{Hag80} (see also
  \cite{HaJo00,HaJo01}): One easily checks that if $A(t)$ and $B(t)$ evolve
  according to \eqref{eq:AB}, then $u$
  given by \eqref{eq:uG} solves \eqref{eq:u}. On the other hand, it is
  clear that \eqref{eq:AB} has a global solution. Finally, since $
  \nabla^2_k E_m\(p(t)\)$ and $\nabla^2_x V\(q(t)\)$ are symmetric
  matrices, it follows from \cite[Lemma~2.1]{Hag80} that for all time,
  $A(t)$ and $B(t)$ satisfy \eqref{eq:ABinv}--\eqref{eq:invBA}.  
\end{proof}
\section{Stability of the approximation up to Ehrenfest
  time} \label{sec:stab} 
  
As a final step we need to show that the derived approximation $\psi_{\rm app}^\eps(t)$ indeed approximates the exact solution $\psi^\eps(t)$ up to 
Ehrenfest time.

\begin{proposition}\label{prop:stab}
 Let Assumptions \ref{hyp:V} and \ref{hyp:E} hold and 
$u_0\in \Sigma^5$.  
Then there exists $C>0$ such that 
\begin{equation*}
\| \psi^\eps (t) - \psi_{\rm app}^\eps(t) \|_{L^2(\R^d)} \le C
\sqrt{\eps} e^{C t},   
\end{equation*}
where $\psi_{\rm app}^\eps(t,x)$ is given by \eqref{eq:app} with
$u(t,z)$ solving 
\eqref{eq:u} subject to $u_{\mid t=0}=u_0$, and $U_0$, $U_1$, $U_2$ are given by
Proposition~\ref{prop:app}. 
\end{proposition}
\begin{proof}
  First, note that $\psi^\eps$ and $\psi^\eps_{\rm app}$ do not
  coincide at time $t=0$, since elliptic inversion has forced us to
  introduce $U_1$ and $U_2$ which are not zero initially. Setting
  $w^\eps= \psi^\eps-\psi^\eps_{\rm app}$, and using \eqref{eq:NLS0} and
  Lemma~\ref{lem:app}, we see that the error solves
  \begin{align*}
    \(i\eps \d_t +\frac{\eps^2}{2}\Delta - 
  V_\Gamma 
  -V \)w^\eps&=
 -\frac{e^{i\phi_m/\eps}}{\eps^{d/4}}
 \eps^{3/2}\(r^\eps+r_1^\eps+r_2^\eps\) (t,z,y)\Big|_{(z,y)= 
    \(\frac{x-q(t)}{\sqrt\eps},\frac{x}{\eps}\) },\\
w^\eps(0,x)&= \eps^{-d/4}\(\sqrt \eps U_1 +\eps
U_2\)\(0,\frac{x-q_0}{\sqrt\eps},\frac{x}{\eps}\) e^{i\phi_m(0,x)/\eps}.
  \end{align*}
The a-priori $L^2$ estimate yields
\begin{align*}
  \|w^\eps(t)\|_{L^2}& \le \sqrt\eps \|U_1(0)\|_{L^2_zL^\infty_y} + 
\eps \|U_2(0)\|_{L^2_zL^\infty_y} \\
&\quad + \sqrt\eps\int_0^t
\(\|r^\eps(s)\|_{L^2_zL^\infty_y} +\|r_1^\eps(s)\|_{L^2_zL^\infty_y}
+\|r_2^\eps(s)\|_{L^2_zL^\infty_y} \)ds. 
\end{align*}
The assertion then follows from Lemma~\ref{lem:app} (establishing the needed 
properties for the functions $r^\eps$, $r^\eps_1$ and $r^\eps_2$),
Proposition~\ref{prop:app}, and Proposition~\ref{prop:growth}. With
this approach, we need to know that $r^\eps$ is in $L^2_z$, so $U_0,U_1$
and $U_2$ have three momenta  in $L^2_z$: in view of
Proposition~\ref{prop:app} and Proposition~\ref{prop:growth}, this
amounts to demanding $u_0\in \Sigma^5$. 
\end{proof}

This asymptotic stability result directly yields the assertion of Theorem~\ref{th:main}. 

\begin{proof}[End of the proof of Theorem~\ref{th:main}]
To conclude, it suffices to notice that 
\begin{equation*}
  \varphi^\eps(t,x) =\psi^\eps_{\rm app}(t,x) -\(\sqrt \eps U_1+\eps
  U_2\)(t,z,y) \Big|_{(z,y)= 
    \(\frac{x-q(t)}{\sqrt\eps},\frac{x}{\eps}\) },
\end{equation*}
so Proposition~\ref{prop:app} and Proposition~\ref{prop:growth} imply
\begin{equation*}
  \|\varphi^\eps(t)-\psi^\eps_{\rm app}(t)\|_{L^2}\le C\sqrt\eps e^{Ct}.
\end{equation*}
This estimate and Proposition~\ref{prop:stab} yield
Theorem~\ref{th:main}. 
\end{proof}
\begin{remark}
  The construction of the approximate solution $\psi^\eps_{\rm app}$
  has forced us to introduce non-zero correctors $U_1$ and $U_2$,
  given by elliptic inversion. Therefore, we had to consider
  well-prepared initial data for $\psi^\eps_{\rm app}$. This aspect is
  harmless as long as one is interested only in the leading order
  behavior of $\psi^\eps$ as $\eps\to 0$. As a consequence, our
  approach would not allow us to construct arbitrary accurate
  approximations for $\psi^\eps$ (in terms of powers of $\eps$),
  unless well-prepared initial data are considered, i.e. data lying in so-called 
  {\it super-adiabatic subspaces}, in the terminology of \cite{PaSpTe03}
  (after \cite{LiBe91}). This is due
  to the spectral analysis implied by the presence of the periodic
  potential $V_\Gamma$, and shows a sharp contrast with the case
  $V_\Gamma=0$.  
\end{remark}

Of course the above given stability result immediately generalizes to
situations where, instead of a single $\varphi^\eps$, a superposition
of finitely many semiclassical wave packets is considered,  
\[
\psi^\eps_0(x) = \eps^{-d/4} \sum_{n=1}^N u_{n} 
\left(\frac{x-q_{n}}{\sqrt\eps}\right) \chi_{m_n}\left(\frac{x}{\eps},
  p_{n} \right) e^{i p_n \cdot (x-q_{n})/ \eps}. 
\]
Since the underlying semiclassical Schr\"odinger equation
\eqref{eq:NLS0} is linear, each of these initial wave packets will
evolve individually from the rest, as in Theorem~\ref{th:main}. 
Up to some technical modifications, it should be possible to consider
even a {\it continuous} superposition of wave packets, yielding a
semiclassical approximation known under the name ``frozen Gaussians'',
see  
\cite{Hel81}.

\bibliographystyle{amsplain}
\bibliography{bloch}
\end{document}